\documentclass[prl,aps,twocolumn,superscriptaddress,a4paper]{revtex4-1}
\usepackage{amsmath,amsthm,amssymb,amsfonts,mathrsfs,amsthm}
\usepackage{dsfont,yfonts,calligra}
\usepackage{graphicx,hyperref}
\usepackage{pxfonts,upgreek}
\usepackage{lmodern}
\usepackage{xcolor}
\DeclareMathAlphabet{\mathpzc}{OT1}{pzc}{m}{it}
\DeclareMathAlphabet{\mathcalligra}{T1}{calligra}{m}{n}

\def\Tr{\operatorname{Tr}} \def\>{\rangle} \def\<{\langle}
\def\N#1{\left|\!\left|{#1}\right|\!\right|} \def\id{\mathsf{id}}
 
 \def\mE{\mathcal{E}}

\def\mN{\mathcal{N}}

\def\sH{\mathcal{H}}  
  
\def\bound{\boldsymbol{\mathsf{L}}}

 \def\dec#1{\mathscr{D}}
 \def\openone{\mathds{1}}

\def\rv#1{{\mathbb{#1}}}

\def\mM{\mathcal{M}}

\renewcommand{\qedsymbol}{\nobreak \ifvmode \relax \else
	\ifdim \lastskip<1.5em \hskip-\lastskip \hskip1.5em plus0em
	minus0.5em \fi \nobreak \vrule height0.75em width0.5em
	depth0.25em\fi}

\renewcommand{\ge}{\geqslant}
\renewcommand{\le}{\leqslant}

\newtheorem{theorem}{Theorem}

\newtheorem{proposition}{Proposition}

\newtheorem{definition}{Definition}

\theoremstyle{remark}

\theoremstyle{definition}

\newcommand{\bea}{\begin{eqnarray}}
\newcommand{\eea}{\end{eqnarray}}
\newcommand{\be}{\begin{equation}}
\newcommand{\ee}{\end{equation}}

\def\N{\mathsf{N}}
\def\D{\mathsf{D}}

\def\sHS{\mathcal{H}_S}

\def\Pr{\operatorname{Pr}}

\begin{document}
	
	%\SetWatermarkText{NOT FOR DISTRIBUTION}
	%\SetWatermarkAngle{60}
	%\SetWatermarkScale{0.5}
	
	%\begin{frontmatter}
	
	\title{\textbf{Noise and disturbance in quantum measurements:
			an information-theoretic approach}}

	\author{Francesco Buscemi}
	
	\affiliation{Institute for Advanced Research, Nagoya University, Chikusa-ku, Nagoya 464-8601, Japan}
	
	\email{buscemi@iar.nagoya-u.ac.jp}
	
	\author{Michael J. W. Hall}
	
	\affiliation{Centre for Quantum Computation and Communication Technology (Australian Research Council), Centre for Quantum Dynamics, Griffith University, Brisbane, QLD 4111, Australia}
	
	\email{michael.hall@griffith.edu.au}
	
	\author{Masanao Ozawa}
	
	\affiliation{Graduate School of Information Science, Nagoya University, Chikusa-ku, Nagoya 464-8601, Japan}
	
	\email{ozawa@is.nagoya-u.ac.jp}
	
	\author{Mark M. Wilde}
	
	\affiliation{Hearne Institute for Theoretical Physics, Department of Physics and Astronomy, Center for Computation and Technology, Louisiana State University, Baton Rouge, Louisiana 70803, USA}
	
	\email{mwilde@lsu.edu}
	
	\begin{abstract}
		We introduce information-theoretic definitions for noise and
		disturbance in quantum measurements and prove a state-independent
		noise-disturbance tradeoff relation that these quantities have to satisfy in
		any conceivable setup. Contrary to previous approaches,
		the information-theoretic quantities we define are invariant under relabelling of outcomes and allow for the possibility of using quantum or classical operations to `correct' for the disturbance. We also show how our bound implies strong tradeoff relations for mean square deviations.
		
	\end{abstract}
	
	\maketitle
	
	% {\footnotesize{\noindent{\bf Keywords:} 
	
	%}}
	
	%\end{frontmatter}

%%%%%%%%%%%%%%%%%%%%%%%%%%%%%%%%%%%%%%%%%%%%%%%%%%%%%%%%%%%%%%%%%%%%%%

%\PRLsection{}

Heisenberg's Uncertainty Principle (HUP) states, loosely speaking,
that in quantum theory a measurement process cannot measure one
observable accurately, such as the position, without causing a measurable disturbance to
another incompatible observable, such as the momentum. 
Notwithstanding the crucial role played by Heisenberg's principle in modern science,
it took a long time between its first exposition~\cite{H27,H83} and
its rigorous formalisation in terms of \emph{noise and disturbance
  operators}~\cite{Oz03,Ozawa-more1,Ozawa-more2}. The statistical spreads of these operators are measurable quantities, and hence tradeoff relations satisfied by these
spreads yield  precise  mathematical translations of Heisenberg's
intuition~\cite{Oz03,Oz04,B13}, that have recently been experimentally tested in a number of scenarios \cite{ESSBOH12,Roz12,WHPWP13,ozawa-exper1,ozawa-exper2,KBOE13,RBBFBW13}.
The use of noise and disturbance operators allows for a detailed,
state-dependent formulation of HUP, able to capture the idea of `how accurate'
 a measurement is with respect to one dynamical variable and `how delicate' the same measurement is with respect to another dynamical variable. %In 

In this paper, we will explore a different approach to HUP,
focused not on the change \emph{per se} in a system's dynamical
variables, but on the \emph{loss of correlation} introduced by this change. In doing
so, we will make use of ideas from information theory such as a
`guessing strategy' and error correction, and our definitions will
be given in terms of information-theoretic quantities like entropies
and conditional entropies.  While we focus on the noise-disturbance context here, our approach also yields tradeoff relations for joint measurements.

In order to understand the difference between the present approach and
the previous one, let us consider, for example, the case of
noise. While the noise, in its conventional form of root-mean-square deviation, is a statistical measure of the  \emph{distance} between a given system observable and the quantity actually measured \cite{Oz04,H04}, here we will only be interested in how well one can \emph{infer} (i.e., guess) the value of a system observable from
a given measurement outcome.  That is to say, we will look only at the degree of correlation between the measurement and the observable, irrespective of how the corresponding outcomes and values are numerically labelled.

Analogously, when characterising the
disturbance, we will consider the measurement process as a source of
noise for the system, and the degree to which such noise can be corrected (for a given observable)
 will give us our definition of disturbance.

Our measures of noise and disturbance will therefore 
quantify the \emph{unavoidable} loss of correlations, i.e., the irreversible components of noise and disturbance. In particular, our definitions and results are invariant under reversible operations, such as relabelling of outcomes and unitary
time evolutions. In contrast, the conventional approach using root-mean-square deviations is not invariant, as such operations can change numerical values and indeed the system observables of interest.

Information-theoretic approaches to HUP-like questions have already
been proposed in a variety of forms~\cite{M07,BHH08,BH2009,WHBH12}.
However, these focus on the disturbance of the system state {\it per se}, % not the case for BRW13
with tradeoff relations which are functions only of the
initial state of the system and the measuring apparatus. In contrast, we define noise and disturbance with respect to two system
observables, and our tradeoff relation depends on the degree to which
such observables are compatible, in the spirit of the original
HUP.  Moreover, our definitions are functions only of the two observables and
the measuring apparatus, leading to a state-independent
tradeoff relation. 

We note that a state-independent noise-disturbance relation has recently been given, for the case of position and momentum observables, in the conventional context of root-mean-square noise and disturbance \cite{BLW13}.  In contrast, our information-theoretic relation applies to arbitrary observables and leads to stronger results.

\textit{Our proposal.}---For simplicity, consider two nondegenerate observables
$X$ and $Z$ of a finite-dimensional quantum system $S$, with corresponding sets of eigenstates $\{|\psi^x\>\}$ and $\{|\varphi^z\>\}$, respectively. The system is subjected to
a measuring apparatus $\mM$.  
Our aim is to introduce an operational context for what it means for
`$\mM$ to measure $X$ accurately,' and for 
`$\mM$ to disturb a subsequent measurement of
$Z$.' This will lead to sensible information-theoretic definitions of \emph{noise} and \textit{disturbance}, $\N(\mM,X)$ and $\D(\mM,Z)$, in terms of operational measurement statistics, which satisfy:

\begin{theorem}\label{theo:main}
  For any measuring apparatus $\mM$ and any non-degenerate observables
  $X$ and $Z$, the following tradeoff between
  noise $\N(\mM,X)$ and disturbance
  $\D(\mM,Z)$ holds
	\begin{equation}\label{eq:tradeoff}
	\N(\mM,X)+\D(\mM,Z)\ge-\log c,
	\end{equation}
	where
$c:=\max_{x,z}|\<\psi^x|\varphi^z\>|^2$ and the $\operatorname{log}$ is in base 2.
\end{theorem}

Theorem~\ref{theo:main} clearly expresses the idea that, whenever observables $X$ and $Z$ are not compatible (i.e., $c<1$), it is impossible to accurately measure one of them without at the same time disturbing the other.  Significant generalizations and applications of this result will be given further below.

In order to proceed, we imagine two corresponding correlation experiments that can be
performed with $\mM$. The first experiment consists of a source producing eigenstates
  $|\psi^x\>$ of $X$ at random; feeding these states into the
  apparatus $\mM$; and determining how correlated the observed outcomes
  $m$ are with the eigenvalues $\xi_x$ of $X$ (Fig.~\ref{fig:noise}). If it is
  possible, from $m$, to guess $\xi_x$ perfectly, then there is
  perfect correlation, and we say that $\mM$ can measure $X$
  accurately, by making the corresponding optimal guess \cite{ [{The state-{\it dependent} definition of precise
      measurement consistent with the present discussion has been given
      in }] Ozawa-perfect1, *Ozawa-perfect2}. Hence, the noise
  $\N(\mM,X)$ will be zero. 
	In general, the
  noise will increase as the probability of correctly guessing $\xi_x$
  decreases. The experiment thus assesses the average performance of
  the apparatus in discriminating between different values of $X$.
  %\vspace{.1in}
\begin{figure}[h]
  	\begin{center}
  		\includegraphics[width=.75\columnwidth]{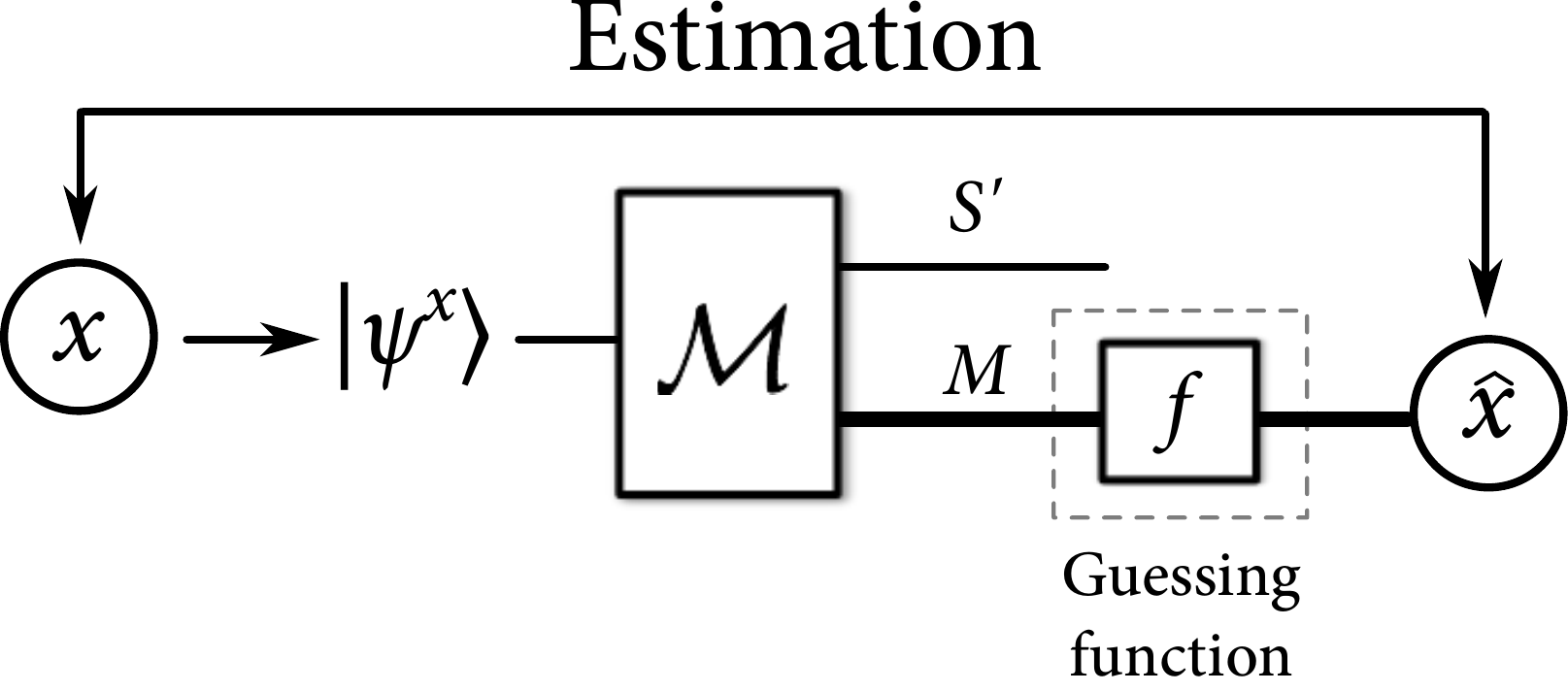}
  	\end{center}
  	\caption{Noise with respect to $X$ is measured by the ability to guess correctly, from the measurement outcome  $M=m$, in which eigenstate $|\psi^x\>$ of $X$ the system was initially prepared. The guessed value, $\hat x$, will in general be some function $f(m)$ of the measurement outcome. Optimisation over $f$ is allowed.}
  	\label{fig:noise}
  \end{figure}
\begin{figure}[h]
  	\begin{center}
  		\includegraphics[width=.75\columnwidth]{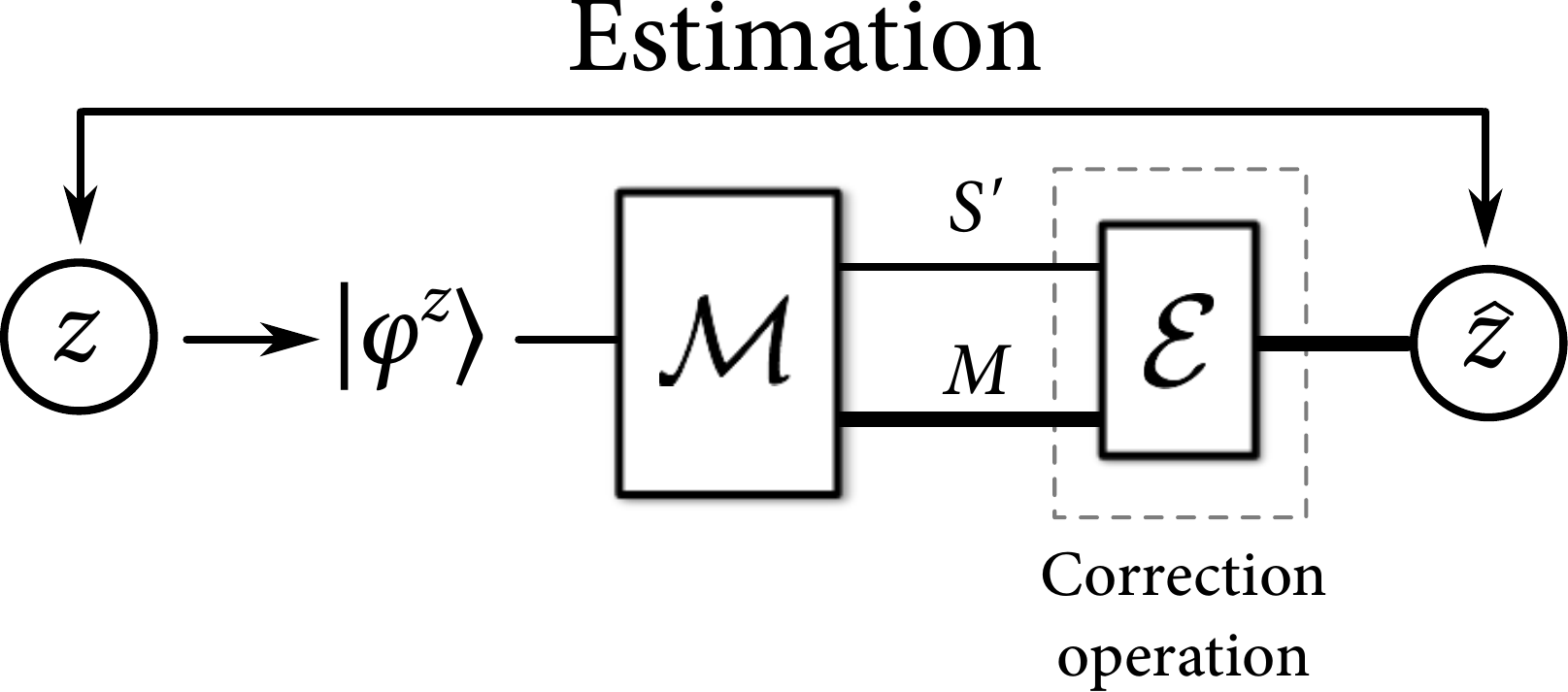}
  	\end{center}
  	\caption{Disturbance with respect to $Z$ is measured by the ability to guess correctly, from the outcome variable $M$ and the `disturbed' output quantum system $S'$, in which eigenstate $|\varphi^z\>$ of $Z$ the system was initially prepared. It is permitted to apply a quantum operation $\mE$ to attempt to correct or minimise any disturbance,  prior to a measurement of $Z$ (with outcome $\hat z$). Optimisation over $\mE$ is allowed.}
  	\label{fig:disturbance}
  \end{figure}
	
In the second experiment, we imagine the source instead producing
  eigenstates $|\varphi^z\>$ of $Z$ at random, and feeding these states through the apparatus $\mM$ (Fig.~\ref{fig:disturbance}). The task is then to
  guess, as accurately as possible, the eigenvalue $\zeta_z$ of the input state $|\varphi^z\>$. We first permit
an arbitrary operation $\mE$, on \textit{both} the classical outcome $m$ and the `disturbed' output quantum system $S'$, to allow for the possibility of `correcting' any reversible disturbance by $\mM$, before making a guess $\hat z$ corresponding to the outcome of a measurement of $Z$. If it is possible to guess perfectly, then there is perfect correlation between $\hat z$ and the input eigenvalue $\zeta_z$, and we say that $\mM$ does not disturb $Z$. Hence, the disturbance $\D(\mM,Z)$ will be zero. In general the disturbance will increase as the correlation decreases.
 % \vspace{.1in}

The main difference between the first and the second correlation experiments is that, in the second one, we are allowed
 to use both the classical outcome observed \emph{and} the output
 quantum system $S'$. This is because  `disturbance' can only be meaningfully defined with
 respect to a measurement of $Z$ that happens \emph{after} the
 measurement process described by $\mM$ has occurred. Thus one is
 allowed to base a guess on all of the data that emerges from the
 measuring apparatus. Alternatively, the second
 experiment can be understood in terms of `error correction': before guessing
 $z$, one tries to `undo,' as accurately as possible, the action of
 the apparatus, seen as a \emph{noisy channel} with both quantum and
 classical outputs.

The notion of disturbance we consider is, therefore, related to the
`irreversible' character of a quantum measurement: any reversible
dynamical evolution is automatically corrected during the
correction stage. 
It therefore captures the
idea of `unavoidable' disturbance, in strong contrast to
the conventional formulation in terms of root-mean-square deviations,
where \emph{any} change in the value of a system's dynamical variables
is considered as a non-trivial disturbance.

\textit{Quantifying noise.}---As discussed above,
we require the information-theoretic noise, $\N(\mM,X)$,
to represent the quality of the correlation between which eigenstate
of $X$ was input and the measurement outcome $m$. For a given input $|\psi^x\>$, this correlation is determined by the conditional probability distribution $p(m|\psi^x)$, which can be measured via the experimental setup in Fig.~1.  Since we are interested in the average noise performance, we have to introduce an
{\it a priori} distribution on the eigenstates $|\psi^x\>$. By fixing
the prior to be the uniform one, i.e., $p(x)=1/d$ where $d$ denotes the dimension of the Hilbert space of the system, 
we obtain the joint input-output probability distribution 
\begin{equation}\label{eq:pmx}
p(m,x)=p(x)\,p(m|\psi^x)=\frac 1d\, p(m|\psi^x).
\end{equation}  
This characterises the correlation between $X$ and $M$, and leads to:

\begin{definition}[Information-Theoretic Noise]\label{def:noise}
  The information-theoretic noise of the instrument $\mM$ as a
  measurement of $X$ is defined as
$\N(\mM,X):=H(X|M)$, where $H(X|M)$ denotes the conditional entropy computed from the joint probability distribution $p(m,x)$ in Eq.~(\ref{eq:pmx}).
\end{definition}

Note that $\N(\mM,X)$ can be interpreted as the average uncertainty as to which eigenvalue of $X$ was input, given the outcome of the measurement.
Our definition can be further justified, in the precise sense that
the noise $\N(\mM,X)$ is small if and only if the
outcome $m$ identifies the eigenvalues of $X$ accurately. In fact, as we show in the Supplemental Material \cite{SM}, standard arguments in information theory~\cite{F61,G68,HR70,HV10} imply that, given a guessing function %$f:M\to\hat X$ 
$\hat x=f(m)$ with its total error probability $p_e:=\Pr\{\hat X\neq X\}=\sum_{x}\sum_{\hat x\neq x}p(\hat x,x)$, it holds that
\begin{equation}\label{eq:fano}
\N(\mM,X)\to 0\ \textrm{iff}\ \min_fp_e\to 0.
\end{equation}

\textit{Quantifying disturbance.}---Let us now consider a second nondegenerate observable
$Z=\sum_z\zeta_z|\varphi^z\>\<\varphi^z|$ of system $S$. While the
noise depends only on the measurement outcome, the disturbance can
depend in principle on both the classical outcome ($M$) and the quantum
output system ($S'$) of $\mM$. However, the conceptual framework 
is analogous to that for noise: we imagine that eigenstates
of $Z$ are acted upon by the measurement process $\mM$. We then
require the \emph{information-theoretic disturbance} to quantify
the extent to which the action of $\mM$ reduces the information about
which eigenstate $|\varphi^z\>$ was initially selected.

However, if we want to quantify truly \emph{unavoidable} disturbance,
we have to allow any possible action aimed at recovering this
information, after the measurement process $\mM$ has taken
place. We therefore allow an optimisation over all possible
correction procedures, before any attempt to estimate $z$ is
conducted. A general correction procedure is modelled by
a completely-positive trace-preserving (CPTP) map $\mE$, reconstructing the initial system $S$ from the output system $S'$ and the measurement record $M$ (Fig.~\ref{fig:disturbance}). The final estimation of $z$ can then be performed via a standard (i.e., von Neumann) measurement of $Z$, since any additional optimisation can be incorporated into the correction channel $\mE$, and no more than $d$ outcomes are needed to discriminate between the input eigenstates. The
information-theoretic disturbance will therefore depend on the joint
probability distribution given by
\begin{equation}\label{eq:pzz}
p(\hat z,z):= p(z)\,p(\hat z|\varphi^z) = \frac 1d \,p(\hat z|\varphi^z),
\end{equation}
which characterises the correlation between $z$ and $\hat z$. In the above equation, as we did before for the case of noise, we are selecting the eigenstates of $Z$ uniformly at random.

We can now formalise the above discussion as follows:
\begin{definition}[Information-Theoretic Disturbance]
\label{def:dist}
  The information-theoretic disturbance that the apparatus $\mM$
  introduces on any subsequent attempt to measure the observable $Z$,
  is defined as
$\D(\mM,Z):=\min_\mE H(Z|\hat Z)$, where the conditional entropy $H(Z|\hat Z)$ is computed from the joint probability distribution $p(\hat z,z)$ in Eq.~(\ref{eq:pzz}), and the minimum is taken over all possible CPTP maps $\mE$.
\end{definition}

As for the information-theoretic noise above, this measure quantifies the average uncertainty of $Z$, given the outcome of the estimate. Eq.~(\ref{eq:fano}) can similarly be applied to justify our definition.
	In fact, besides Eq.~(\ref{eq:fano}), the notion of disturbance we have introduced can be given an alternative interpretation, directly related to the idea that the measurement process irreversibly disturbs the measured system. Defining the probability of error as $p_e=\sum_{z}\sum_{\hat z\neq z}p(\hat z,z)$, the probability of  guessing correctly, $1-p_e$, is nothing but the average fidelity of correction, i.e.,
	$1-p_e=d^{-1}\sum_z\operatorname{F}\big\{(\mE\circ\mM)(|\varphi^z\>\<\varphi^z|)\ ,\  |\varphi^z\>\<\varphi^z|\big\}$,
	where $\operatorname{F}\{\rho ,\sigma\}:=\Tr\left[\sqrt{\sqrt{\rho}\sigma\sqrt{\rho}}\right]^2$ is the fidelity between states $\rho$ and $\sigma$ \cite{U73,J94}.

\textit{Information-theoretic noise-disturbance relation.}---Defs.~\ref{def:noise} and~\ref{def:dist} lead  to the noise-disturbance relation~(\ref{eq:tradeoff}), as shown in the Supplemental Material \cite{SM}.  The proof is based on a mapping of the statistics of the two estimation procedures in Figs.~\ref{fig:noise} and~\ref{fig:disturbance} (which require separate inputs of eigenstates of $X$ and $Z$), to the measurement statistics of a single maximally-entangled state, and applying the Maassen-Uffink entropic uncertainty relation~\cite{Maassen} to this state.

\textit{Useful quantum lower bound on disturbance.}---Given an apparatus $\mM$ and two observables $X$ and $Z$, both the noise
$\N(\mM,X)$ and disturbance $\D(\mM,Z)$ can in principle be
computed. However, while the noise can be computed directly from the
data, the definition of disturbance involves an optimisation over all possible
correction procedures. Such an optimisation can in general be very
hard to perform. Such a problem can however be encompassed simply by
noticing that any correction
followed by an estimation of $Z$ is nothing but a post-processing of
systems $S'$ and $M$ into the estimated variable $\hat Z$. We can therefore apply the quantum
data-processing inequality~\cite{LR73,U77,SN96} to arrive at the
lower bound $\D(\mM,Z)\ge H(Z|S'M)$ for the disturbance, where the conditional \textit{quantum} entropy $H(A|B)$ is defined as the difference of the von Neumann entropies corresponding to  the combined system $AB$ and system $B$. As we prove in the Supplemental Material \cite{SM}, it is also possible to refine the bound of Theorem~\ref{theo:main} to
\begin{equation}\label{eq:quantum-bound}
\N(\mM,X)+H(Z|S'M)\ge-\log c.
\end{equation}
 The idea is that one can use the Stinespring representation theorem~\cite{S55} to purify the action of the measurement channel $\mM$ to an isometry $V:\sH_S\to\sH_{S'}\otimes\sH_M\otimes\sH_E\otimes\sH_{\bar M}$, where the additional systems $E$ and $\bar M$ represent the environment and the environment's redundant copy of $M$, respectively, and then apply the recently discovered complementarity relations in the presence of quantum memory~\cite{BCCRR10,CYGG11}.

\textit{Generalisations.}---Since the proof of the tradeoff relation (\ref{eq:tradeoff}) only requires that $\hat X$ and $\hat Z$ are joint estimates of $X$ and $Z$ \cite{SM}, it follows as an immediate corollary that 
\begin{equation}\label{eq:intrinsic-Hall}
\N(\mM,X) + \N(\mM,Z) \ge -\log c
\end{equation}
for any joint estimate of $X$ and $Z$ via measurement apparatus $\mM$. This information-theoretic {\it joint-measurement} tradeoff relation  may be contrasted to the joint-measurement information exclusion relation of Hall \cite{H97}.  In particular, contrarily to the latter, Eq.~(\ref{eq:intrinsic-Hall}) is state \emph{independent}, i.e, it constrains the {\it inherent} degree to which the measurement apparatus can simultaneously perform as an $X$-measuring device and a $Z$-measuring device.

Moreover, we can generalise the tradeoff relation in the theorem to degenerate observables $X$ and $Z$. One can use essentially the same arguments as in the non-degenerate case, with a suitable entropic uncertainty relation~\cite{KP02}, to replace the constant $c$ in Theorem~\ref{theo:main} with $c'=\max_{x,z} \|X_xZ_z\|^2_\infty$, where $X_x$ and $Z_z$ are the \textit{spectral projectors} corresponding to distinct eigenvalues of $X$ and $Z$, respectively (see Supplemental Material \cite{SM}).

Our information-theoretic approach also yields tradeoff relations for root-mean-square deviations.  In particular, for the joint probability distribution $p(m,x)$ in Eq.~(\ref{eq:pmx}), consider the alternative measure of noise defined by $V_{\N}^X := \sum_{m,x} p(m,x) [\hat x - \xi_x ]^2$,
where $\hat x=f(m)$ is the estimate of $x$ from measurement outcome $m$.  Thus, this measure is just the mean square deviation (MSD) of the estimate of the input eigenvalue from its true value.  One may similarly define a measure of disturbance by the MSD
$V_{\D}^Z := \sum_{\hat z,z} p(\hat z,z) [\hat z -\zeta_z]^2$.
As shown in the Supplemental Material \cite{SM}, these quantities are equal to squares of the noise and disturbance measures defined by Ozawa \cite{Oz03}, for the particular case of a maximally-mixed system state $\rho_S=d^{-1}\openone_S$.  If the spacing between eigenvalues of $X$ ($Z$) is a multiple of some value $s_X$ ($s_Z$), then the tradeoff relation 
\begin{equation} \label{variance}
\left[ V_{\N}^X + \frac{(s_X)^2}{12} \right] \,  \left[ V_{\D}^Z + \frac{(s_Z)^2}{12} \right] \ge \left( \frac{s_X s_Z}{2\pi e \,c}\right)^2, 
\end{equation}
follows as a corollary to Theorem~\ref{theo:main} (see Supplemental Material \cite{SM}). It follows, for example, that if $X$ and $Z$ are the Pauli spin operators for a qubit system, then $ V_{\N}^X$ and $V_{\D}^Z$ cannot both vanish.  This cannot, in contrast, be concluded from known noise-disturbance tradeoff relations for the maximally-mixed state \cite{Oz03,B13}.

A generalisation to continuous observables is not straightforward operationally, as the corresponding eigenkets are not physical states.  However, as described in the Supplemental Material \cite{SM}, it is possible to formally take limits to obtain the tradeoff relation
\begin{equation}\label{eq:PQ-case}
\N(\mM,Q) + \D(\mM,P) \ge \log \pi e\hbar
\end{equation}
for position $Q$ and momentum $P$. Moreover, defining MSDs $V_{\N}^Q$ and $V_{\D}^P$ as above, this bound further implies the Heisenberg-type noise-disturbance relation
$V_{\N}^Q\,V_{\D}^P \ge \hbar^2/4$,
in the same way that the usual Heisenberg uncertainty relation follows from the entropic uncertainty relation for $Q$ and $P$ \cite{BM75}. Note this is similar in form to, but stronger than, the relation recently obtained by Busch {\it et al.} \cite{BLW13}, as the latter is for the product of the maximum possible deviations, %(which can easily diverge \cite{Ozcom})
rather than for the product of the mean deviations. In both cases, however, the measures of noise and disturbance for position and momentum are purely formal, with no operational counterparts. Hence it appears that {\it state-dependent} noise-disturbance and joint-measurement relations \cite{H97,Oz03,Oz04,H04,Ozjoint,B13,WHPWP13} may be preferable for continuous observables.

Finally, while Theorem~\ref{theo:main} relies on the entropic uncertainty relation due to Maassen and Uffink, any such relation, such as those recently obtained by Puchala  {\it et al.} \cite{PRZ13} and Coles {\it et al.} \cite{CP13},  will similarly lead to a corresponding tradeoff relation for the information-theoretic noise and disturbance.  

\textit{Conclusion.}---We have obtained an information-theoretic characterisation of Heisenberg's Uncertainty Principle, which for the first time characterises the inherent degree to which a given measurement apparatus must disturb one observable to gain information about another observable, independently of the state of the system undergoing measurement.  Our proposed measures of noise and disturbance quantify the irreversible loss of correlations introduced by the measurement apparatus, and are invariant under operations such as relabelling of outcomes and invertible evolutions. Further, in the case of discrete observables, they can be operationally determined, as per Figs.~\ref{fig:noise} and~\ref{fig:disturbance}. %Further, the minimal disturbance can be interpreted as an average fidelity of correction, and we have shown that both it and our tradeoff relation are bounded below by a quantum conditional entropy.
Our main theorem has a number of generalisations, including extensions to tradeoff relations for joint measurements, degenerate observables, root-mean-square deviations and continuous observables, and yields stronger constraints than previous results in the literature.  %Future work includes ......

We believe the  above fundamental results will have diverse applications in quantum information theory and quantum metrology, for example, and will also motivate experimental confirmation of the strong information-theoretic form of the Heisenberg Uncertainty Principle presented here.  

\textit{Acknowledgements.}---We acknowledge R. Colbeck and G. Smith for helpful discussions, and thank A. Rastegin for pointing out an issue in an early version of the proof of Theorem~\ref{theo:main}. FB is supported by the Program for Improvement
of Research Environment for Young Researchers from SCF
commissioned by MEXT of Japan.
MJWH is supported by the ARC Centre of Excellence CE110001027. MO  is supported by the John Templeton Foundations, ID \#35771,
JSPS KAKENHI, No. 21244007, and MIC SCOPE, No.121806010. MMW acknowledges support from the
Centre de Recherches Math\'ematiques and from the Department of Physics and Astronomy at LSU.

\appendix

\bibliography{Ref}

%\end{document}

\newpage
\onecolumngrid

\setcounter{page}{1}
\renewcommand{\thepage}{Supplemental Material -- \arabic{page}/8}
\setcounter{equation}{0}
\renewcommand{\theequation}{S.\arabic{equation}}

\section{Notation and quantum instruments}

In what follows, we restrict all classical random variables ($X$, $Z$,
$M$, etc) to have finite ranges ($\rv{X}=\{x\}$,
$\rv{Z}=\{z\}$, $\rv{M}=\{m\}$, etc), and we restrict all quantum systems ($S$,
$E$, etc) to be associated with finite-dimensional Hilbert spaces
($\sH_S$, $\sH_E$, etc).  %The set of linear operators acting on a
%Hilbert space $\sH$ is denoted by $\bound(\sH)$, while density
%operators (i.e., non-negative, trace-one linear operators) are denoted
%by Greek letters like $\rho$ and $\sigma$. Completely positive (CP)
%maps are denoted as, e.g., $\mE:\bound(\sH_{S})\to\bound(\sH_{S'})$;
%if, in addition, the map is also trace-preserving (TP), we use the
%term \emph{channel}. Finally, the identity operator in $\bound(\sH)$ is denoted by $\openone$, while the identity \emph{channel} is denoted by $\id$.

The apparatus $\mM$ may in general be represented by a quantum
instrument~\cite{DL70,D76,Ozawa1984}. Recall that a quantum instrument
is a collection $\{\mM^m\}_{m\in\rv{M}}$ of completely positive (CP) maps
$\mM^m$, each mapping linear operators on the input space $\sH_S$ into linear operators on the output space $\sH_{S'}$,
%:\bound(\sH_S)\to\bound(\sH_{S'})$ 
and such that the sum  map,
$\sum_m\mM^m$, is trace-preserving (TP), i.e., a channel.
The interpretation of a
quantum instrument as a measurement process goes as follows: given the
initial state $\rho$ of the system $S$ undergoing the measurement, the
$m$th outcome is observed with probability $p(m)=\Tr[\mM^m(\rho_S)]$,
in which case the measuring apparatus will return an output quantum
system $S'$ in state $\sigma^m_{S'}=\mM^m(\rho_S)/p(m)$. Any measurement process can be accommodated in the quantum instrument formalism~\cite{Ozawa1984}.

%\PRLsubsection*{Instruments as channels}
%
%The object which is given is a quantum instrument
%$\mM=\{\mM^m\}_{m\in\mathbb{M}}$, mapping an input quantum system $S$ (the object of the measurement) to a classical random variable $M$ (with range $\rv{M}$ and representing the visible \emph{outcome} of the measurement) and an output quantum system $S'$ carrying the so-called \emph{reduced} quantum state.

A convenient way to study quantum instruments from an
information-theoretic viewpoint is to represent the collection $\{\mM^m\}_{m\in\rv{M}}$
as a single CPTP map, also denoted by $\mM$, acting as follows:
\begin{equation}\label{eq:I-channel}
\mM(\rho_S)=\sum_m\mM^m_S(\rho_S)\otimes|m\>\<m|_M,
\end{equation}
where $M$ is now an additional quantum system encoding the classical  measurement outcome on 
orthonormal (and, therefore, perfectly distinguishable) `flag' states 
$|m\>\<m|$. In what follows, therefore, the symbol $\mM$ will be used to denote both the apparatus and the channel in~(\ref{eq:I-channel}).

The conditional input-output probability
distribution $p(m|\psi^x)$ for measurement outcome $m$ given input eigenstate $|\psi^x\>$ is independent of the quantum output $S'$, and is given by $p(m|\psi^x):=\Tr[\mM^m(|\psi^x\>\<\psi^x|)]$. Hence the joint probability distribution $p(x,m)$ in Eq.~(\ref{eq:pmx}) of the main text is given by
\begin{equation}\label{eq:pmxSM}
p(m,x)=\frac 1d \Tr[\mM^m(|\psi^x\>\<\psi^x|)].
\end{equation} 
Moreover, the joint probability distribution $p(x,m)$ in Eq.~(\ref{eq:pzz}) of the main text is given by
\begin{equation}\label{eq:pzzSM}
p(\hat z,z)=\frac 1d \, \<\varphi^{\hat z}|(\mE\circ\mM)(|\varphi^z\>\<\varphi^z|)|\varphi^{\hat z}\>.
\end{equation}

\section{Formalisation of Eq.~(\ref{eq:fano})}

Here we formalise the content of Eq.~(\ref{eq:fano}) as follows:

\begin{proposition}\label{prop:1}
There exists a guessing function $f:M\to\hat X$ such that 
\[
p_e\le\frac 12\N(\mM,X),
\]
where $$p_e:=\Pr\{\hat X\neq X\}=\sum_{x}\sum_{\hat x\neq x}p(\hat x,x)$$ is the total
error probability. On the other hand, for any guessing function
$f:M\to\hat X$,
\[
\N(\mM,X)\le h(p_e)+p_e\log(|\rv{X}|-1), 
\]
where $h(p_e)=-p_e\log p_e-(1-p_e)\log(1-p_e)$ is the binary entropy
computed from $p_e$. Therefore
\begin{equation*}
\N(\mM,X)\to 0\qquad\Longleftrightarrow\qquad p^{\operatorname{min}}_e:=\min_f\Pr\{\hat X\neq X\}\to 0.
\end{equation*}
\end{proposition}
\begin{proof}
  For the first relation, see Gallager~\cite{G68}, 
  Hellman~\cite{HR70}, or Ho and Verd\'u~\cite{HV10}. The second
  relation is Fano's inequality~\cite{F61}.
\end{proof}

\section{Proof of Theorem~\ref{theo:main}}

Both noise $\N(\mM,X)$ and disturbance $\D(\mM,Z)$ are defined in terms of optimal guessing strategies: for the noise, we check how well one can guess the value of $X$ from the measurement outcome $M$, while for the disturbance, we look instead at how well one can guess the value of $Z$ from $M$ and $S'$ together. As these tasks can be performed simultaneously, since we can copy the classical random variable $M$, we can actually reformulate our two correlation experiments as a joint estimation of two observables, $X$ and $Z$, by means of a single observation. As a consequence, we can directly apply known entropic uncertainty relations \cite{BM75,Maassen,Hall1995,KP02,PRZ13,CP13} and find a lower bound on the sum of $\N(\mM,X)$ and $\D(\mM,Z)$, as anticipated by Theorem~\ref{theo:main}.

\begin{proof}
First, we introduce an auxiliary quantum system $R$ (i.e., a \emph{reference} copy), with a Hilbert space isomorphic to that of the input system $S$, i.e., $\sH_R\cong\sH_S$. We then fix  orthonormal basis sets $\{|i\>_S\}$ and $\{|i\>_R\}$, for  $\sH_S$ and $\sH_R$ respectively, and define the maximally entangled state
	\[
	|\Phi^+_{RS}\>:=\frac{1}{\sqrt{d}}\sum_{i=1}^{d}|i_R\>\otimes|i_S\>.
	\]
	By direct inspection, it is easy to check that, for any linear operator $X\in\bound(\sH_R)$,
	\[
	\Tr_R[(X_R\otimes\openone_S)\ |\Phi^+_{RS}\>\<\Phi^+_{RS}|]=\frac 1d X_S^T,
	\]
	where the exponent $T$ denotes the transposition made with respect to the chosen basis $\{|i_S\>\}$. Therefore, we have the so-called `ricochet' property
	\[
	\frac 1d|\psi^x_S\>\<\psi^x_S|=\Tr_R[(|\psi^x_R\>\<\psi^x_R|^T\otimes\openone_S)\ |\Phi^+_{RS}\>\<\Phi^+_{RS}|],
	\]
	and analogously for $|\varphi^z_S\>\<\varphi^z_S|$, relating observables on $S$ to their transpositions on $R$.
	
	We then notice that (as anticipated above) the two correlation experiments defining noise and disturbance can be viewed as a \emph{single} estimation producing a pair of random variables $U=(W,W')$. Here, we may take $W$ to be just a copy of $M$, while  $W'$ is the best possible estimate $\hat Z$ for $Z$ (corresponding to an optimal correction operation $\mE$ in Fig.~\ref{fig:disturbance}). Let therefore $\{\Pi^{u}\}_{u\in\rv{U}}$ be the POVM corresponding to the estimation of $U$, so that $p(u,x)=\frac 1d\Tr[\Pi^{u}\ |\psi^x_S\>\<\psi^x_S|]$ and $p(u,z)=\frac 1d\Tr[\Pi^{u}\ |\varphi^z_S\>\<\varphi^z_S|]$. Notice that the POVM $\{\Pi^{u}\}_{u\in\rv{U}}$ incorporates any and all preceding correction procedures (i.e., the action of the correction channel $\mE$ and optimal guessing function). Exploiting the `ricochet' property explained above, we therefore arrive at
	\[
	p(u,x)=\Tr[(|\psi^x_R\>\<\psi^x_R|^T\otimes\Pi^{u}_S)\ |\Phi^+_{RS}\>\<\Phi^+_{RS}|],
	\]
	and
	\[
	p(u,z)=\Tr[(|\varphi^z_R\>\<\varphi^z_R|^T\otimes\Pi^{u}_S)\ |\Phi^+_{RS}\>\<\Phi^+_{RS}|].
	\]
Note this property also implies that the operational joint distribution $p(m,x)$, used to define the information-theoretic noise in Definition~\ref{def:noise},  may alternatively be obtained by measuring $X^T$ on the reference $R$ and $\mM$ on the system $S$ (and similarly for the information-theoretic disturbance).

Now, by defining the ensemble of reference states ${\cal U}\equiv \{\rho^{u}_R; p(u)\}$ via
	\[
	p(u)\rho^{u}_R:=\Tr_S[(\openone_R\otimes\Pi^{u}_S)\ |\Phi^+_{RS}\>\<\Phi^+_{RS}|],
	\]
	we have $p(u,x)=p(u)\Tr[|\psi^x_R\>\<\psi^x_R|^T\ \rho^{u}_R]$ and $p(u,z)=p(u)\Tr[|\varphi^z_R\>\<\varphi^z_R|^T\ \rho^{u}_R]$. We  then find, using the Maassen-Uffink entropic uncertainty relation \cite{Maassen}, that
	\[ H(X|U)+H(Z|U) = \sum_{u} p(u) \left[ H_{\rho^{u}_R}(X^T)+H_{\rho^{u}_R}(Z^T) \right] \ge -\log c', \]
	where $H_\rho(A)$ is the entropy of observable $A$ for state $\rho$, and $c':=\max_{x,z}\Tr[|\psi^x_R\>\<\psi^x_R|^T\ |\varphi^z_R\>\<\varphi^z_R|^T]$. Using invariance of the trace under transposition of its argument, $c'=c$, and the (classical) data-processing inequality, i.e.,  $H(X|U)=H(X|M\hat Z)\le H(X|M)=\N(\mM,X)$ and $H(Z|U)=H(Z|M\hat Z)\le H(Z|\hat Z)=\D(\mM,Z)$, we finally obtain Eq.~(\ref{eq:tradeoff}).
\end{proof}

\section{Proof of Eq.~(\ref{eq:quantum-bound})}

\label{app:UPPQM}

As said above, in what follows we will describe quantum measurement
processes as channels (i.e., CPTP maps)
$\mM:\bound(\sHS)\to\bound(\sH_{S'})\otimes\bound(\sH_{M})$.

A fundamental tool to study channels from an information-theoretic
viewpoint is provided by the so-called \emph{Stinespring isometric
	dilation}, that is, an ancillary system together with an isometry
	$V$ (i.e., such that $V^\dag V=\openone$), in terms of which the given
	channel can be obtained as an open evolution by tracing over the
	ancillary degrees of freedom~\cite{S55}. It is known that, whenever
	the output of a channel contains a classical register, like the system
	$M$ in Eq.~\eqref{eq:I-channel}, the ancillary system holds a
	perfectly correlated copy of such a register. Therefore, in our case,
	the Stinespring isometric dilation will be an operator $V: \sH_{S} \to
	\sH_{S'}\otimes\sH_{M}\otimes\sH_{E}\otimes\sH_{\bar M}$, where the
	ancillary systems $E$ and $\bar M$ play the roles, respectively, of
	the environment (which can be thought of as comprising also the probe's degrees of freedom, according to the indirect measurement model~\cite{Oz04}) and of the environment's perfect copy of $M$.
	
	\begin{figure}[h]
		\begin{center}
			\includegraphics[width=8cm]{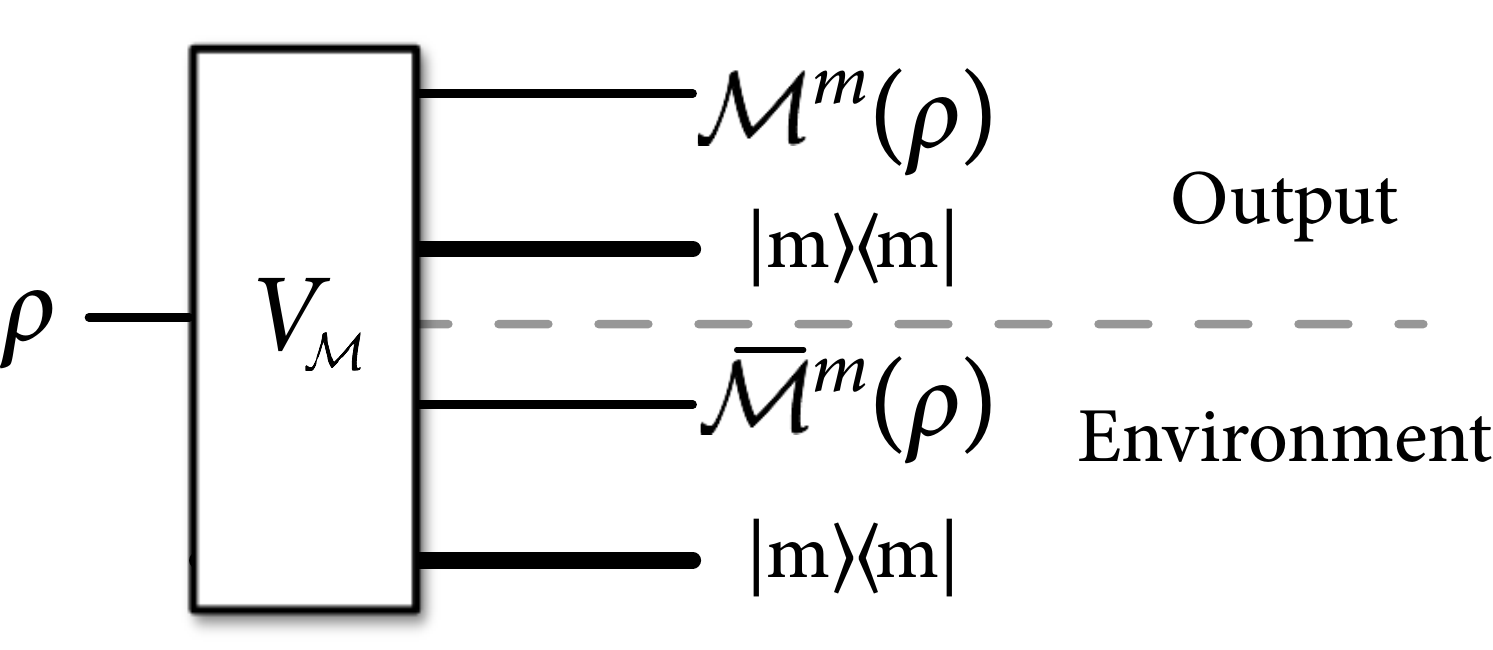}
			\end{center}
			\caption{Stinespring isometric dilation: the environment gets full
				information about the output classical random variable. The CP maps
				$\bar\mM^m$ denote the quantum output to the environment.}
				\label{fig:stine}
				\end{figure}

%In order to lower bound the sum of noise and disturbance, we will
%employ the following trick:
%
%\begin{lemma}\label{lemma:1}
%	Let $\{|h^i\>\}_{i=1}^d$ be a preferred basis in the Hilbert space
%	$\sHS$, let $\sHR$ be a reference copy of $\sHS$
%	(i.e. $\sHR\cong\sHS$), and let $|\Phi^+\>\in\sHR\otimes\sHS$ the
%	maximally entangled state defined as
%	\[
%	|\Phi^+\>:=\frac 1{\sqrt d}\sum_{i=1}^d|h^i_R\>\otimes|h^i_S\>.
%	\]
%	Then, for any linear operator $L\in\bound(\sHR)$,
%	\[
%	\Tr_R[(L_R\otimes\openone_S)\ |\Phi^+_{RS}\>\<\Phi^+_{RS}|]=\frac 1d L^T_S,
%	\]
%	where the superscript $T$ denotes the transposition with respect to
%	the basis $\{|h^i\>\}$.
%\end{lemma}
%\begin{proof}
%	By direct inspection:
%	\[
%	\begin{split}
%	\Tr_R[(L_R\otimes\openone_S)\
%	|\Phi^+_{RS}\>\<\Phi^+_{RS}|]&=\frac 1d\sum_{ij}\Tr_R[(L_R\otimes\openone_S)\
%	|h^i\>\<h^j|_R\otimes |h^i\>\<h^j|_R]\\
%	&=\frac 1d\sum_{ij}\left\{\<h^j_R|L_R|h^i_R\>\right\}\ |h^i\>\<h^j|_S\\
%	&=\frac 1d\sum_{ij}\left\{\<h^i_R|L_R^T|h^j_R\>\right\}\ |h^i\>\<h^j|_S\\
%	&=L_S^T.
%	\end{split}
%	\]
%\end{proof}

The proof proceeds as follows: as already done in proving Theorem~\ref{theo:main}, we introduce an auxiliary reference system $R$, with $\sH_R\cong\sH_S$, and define the maximally entangled state 
\[
|\Phi^+_{RS}\>=\frac{1}{\sqrt{d}}\sum_i|i_R\>\otimes|i_S\>.
\]
Then, the Stinespring isometry
	$V:\sHS\to\sH_{S'}\otimes\sH_{M}\otimes\sH_E\otimes\sH_{\bar M}$ providing a
	dilation of the channel $\mM$ is applied on $\sHS$ alone, resulting in
	\[
	|\Omega_{RS'ME\bar M}\>:=(\openone_R\otimes V_S) |\Phi^+_{RS}\>.
	\]
	Notice that, by definition of Stinespring isometry, $\Tr_{E\bar M}\left[|\Omega_{RS'ME\bar M}\>\<\Omega_{RS'ME\bar M}|\right]
	=(\id_R\otimes\mM_S)(|\Phi^+_{RS}\>\<\Phi^+_{RS}|)$, i.e. the state $|\Omega_{RS'ME\bar M}\>$ is a purification of the mixed state
	\[
	\omega_{RS'M}:=\sum_m(\id_R\otimes\mM^m_S)(|\Phi^+_{RS}\>\<\Phi^+_{RS}|)\otimes|m\>\<m|_M.
	\]
	Then, while the system $E$ purifies systems $R$ and $S'$, the system $\bar M$ purifies the `classical outcome' system $M$, i.e. the joint pure state $|\Omega_{RS'ME\bar M}\>$ can always be thought of as being in the following form:
	\[
	|\Omega_{RS'ME\bar M}\>=\sum_m|\Omega^m_{RS'E}\>\otimes|m_M\>\otimes|m_{\bar M}\>,
	\]
	where $\Tr_{E}[|\Omega^m_{RS'E}\>\<\Omega^m_{RS'E}|]=(\id_R\otimes\mM^m_S)(|\Phi^+_{RS}\>\<\Phi^+_{RS}|)$ for all $m$. From the above equation it is also clear the reason why the additional system $\bar M$ is automatically perfectly correlated with $M$. (In fact, there is an extra unitary degree of freedom on the purifying systems $E$ and $\bar M$; however, since we are interested in entropic quantities that are unitarily invariant, we can, without loss of generality, work with the particular Stinespring isometry given above.)
	
	We now apply the entropic uncertainty relation in the presence of quantum memory~\cite{BCCRR10} for pair of observables $X$ and $Z$, replacing the tripartition $A-B-E$ appearing in Corollary~2, Supplementary Information of Ref.~\cite{BCCRR10}, with the tripartition $R-S'M-E\bar M$ we have here; doing so, we directly obtain the following relation:
	\begin{equation}\label{eq:berta}
	H(Z|S'M)+H(X|E\bar M)\ge-\log c
	\end{equation}
(see also Corollary 5 of Ref.~\cite{CYGG11} for a related result).	Notice that in the formula above both $H(Z|S'M)$ and $H(X|E\bar M)$ are conditional quantum entropies. Eq.~(\ref{eq:quantum-bound}) is finally recovered from~(\ref{eq:berta}) by quantum data-processing inequality, i.e. $H(X|E\bar M)\le H(X|\bar M)=H(X|M)$, where the last equality holds because $M$ and $\bar M$ are perfectly correlated.

\section{Case of degenerate observables}

As anticipated in the main text, we can easily generalise the tradeoff relation of Theorem~\ref{theo:main} to degenerate observables $X$ and $Z$. However, in this case, the correlation experiments need to be slightly modified as follows. First, an eigenvalue $\xi_x$ of $X$ is drawn at random according to the \textit{a priori} distribution $p(x)=d_x/d$, where $d_x$ is the degeneracy of $\xi_x$. This means that eigenvalues with larger degeneracy are more likely to be chosen: this is a very natural requirement if we want to give a state-independent characterisation of a given measuring apparatus. Then, conditional on the eigenvalue chosen, a pure state $|\psi^x\>$ is drawn at random, according to the uniform (Haar) distribution on the corresponding eigenspace. This is also a very natural assumption for us, since we are interested in evaluating the average performance of the measuring apparatus. This means that, in the presence of degeneracy, the effective state input through $\mM$ in Fig.~\ref{fig:noise} is represented by $X_x/d_x$, where $X_x$ denotes the orthogonal projector onto the eigenspace corresponding to eigenvalue $\xi_x$. For operational purposes, it is also worth remarking that the same effective input state, for eigenvalue $x$, is obtained by transmitting any of a set of mutually orthogonal pure eigenstates with equal prior probabilities.

Similarly, if eigenvalue $\zeta_z$ of $Z$ has degeneracy $d'_z$, then the protocol in Fig.~\ref{fig:disturbance} is modified to input states $Z_z/d'_z$ with prior probability $d'_z/d$, where $Z_z$ denotes the orthogonal projector onto the eigenspace corresponding to eigenvalue $\zeta_z$. 

The proof of Theorem~\ref{theo:main} may now be followed as before, but using the Krishna and Parthasarathy entropic uncertainty relation for degenerate observables \cite{KP02}, to give the tradeoff relation (\ref{eq:tradeoff}) with $c$ generalised to 
\[c' = \max_{x,z} \|X_xZ_z\|^2_\infty, \]
where $X_x$ and $Z_z$ denote the projection operators onto the eigenspaces of $X$ and $Z$ for eigenvalues $\xi_x$ and $\zeta_z$ respectively. We recall here that, for any linear operator $A$, its \textit{infinity norm} $\|A\|_\infty$ is defined as the largest of its singular values. The infinity norm therefore coincides with the standard operator norm defined as $\|A\|=\max_{v:\|v\|=1}\|Av\|$. Moreover, in the case $X_x=|\psi^x\>\<\psi^x|$ and $Z_z=|\varphi^z\>\<\varphi^z|$, one has $\|X_xZ_z\|^2_\infty=|\<\psi^x|\varphi^z\>|^2$, i.e., the non-degenerate case is recovered.

As a final remark, we notice that, in the presence of degeneracy, the interpretation of disturbance in terms of average fidelity of correction (as we presented it after Definition~\ref{def:dist}) becomes problematic. It is however still true that our measure of disturbance quantifies how much information about eigenvalue $\zeta_z$ is left available after the measurement $\mM$ has been performed.

\section{Noise and disturbance relations for mean square deviations} \label{app:msd}

Here the noise-disturbance relation (\ref{variance}) for mean square deviations 
\begin{equation}
V_{\N}^X := \sum_{m,x} p(m,x) [\hat x - \xi_x]^2 ,~~~~ V_{\D}^Z := \sum_{\hat z,z} p(\hat z,z) [\hat z -\zeta_z]^2 
\end{equation}
 is derived, and it is shown that these quantities are equal to the conventional, root-mean-square based notions of noise and disturbance for the case of a maximally mixed system state.  

First, notice that if a random variable $N$ takes integer values, then the entropy of $N$ may be bounded in terms of its variance by \cite{Mow}
\[ H(N)  \le \frac{1}{2} \log \left\{2\pi e\left[{\rm Var} N+\frac{1}{12}\right]\right\} . \]
Hence, if the spacing between eigenvalues of $X$ is a multiple of $s_X$, i.e., $\xi_x = \xi+s_X n_x$ for some fixed $\xi$ and integer $n_x$, then it immediately follows that
\begin{equation} \label{entvar}
H(X) \le \frac{1}{2} \log\left\{ 2\pi e\left[{\frac{1}{s_X^2} \rm Var} X +\frac{1}{12}\right] \right\}
\end{equation}
using $H(X)=H(N)$ and ${\rm Var} X=s_X^2{\rm Var} N$, where $N$ ranges over the $n_x$.

It follows, writing the estimate of $X$ as $\hat x=f(m)$ (see main text), and defining $\bar{x}_m:=\sum_x x\,p(x|m)$, that
\begin{align*}
V_{\N}^X+\frac{s_X^2}{12}  &:= \sum_{m,x} p(m,x)\,[f(m)-x]^2 +\frac{s_X^2}{12}\\
&= \sum_m p(m)\, \sum_x p(x|m) [f(m)-x]^2 +\frac{s_X^2}{12}\\
&\ge \sum_m p(m)\, \sum_x p(x|m) (x-\bar{x}_m)^2+\frac{s_X^2}{12}\\
&= \sum_m p(m)\, \left[{\rm Var}_m X+\frac{s_X^2}{12}\right]\\
& \ge \frac{s_X^2}{2\pi e} \sum_m p(m)\,2^{2H_m(X)}\\
&\ge \frac{s_X^2}{2\pi e} 2^{\left[ 2\sum_m p(m) H_m(X) \right]}\\
&= \frac{s_X^2}{2\pi e} 2^{2H(X|M)}
\end{align*}
where ${\rm Var_m X}$ and $H_m(X)$ denote the variance and entropy of $X$ for a fixed value of $m$, and the property $\langle (A-\alpha)^2\rangle\ge {\rm Var}A$ for any random variable $A$ and the convexity of $2^x$ have been used.

In precisely the same manner, one also has the inequality
\[ V_{\D}^Z +\frac{s_Z^2}{12} \ge  \frac{s_Z^2}{2\pi e} 2^{2H(Z|\hat Z)}, \]
when the spacing between the eigenvalues of $Z$ is a multiple of $s_Z$.  Combining this with the above inequality then gives
\begin{align} \nonumber
\left[V_{\N}^X+\frac{s_X^2}{12} \right]\,\left[ V_{\D}^Z +\frac{s_Z^2}{12} \right]&\ge  \left(\frac{s_X s_Z}{2\pi e}\right)^2 2^{2H(X|M)+2H(Z|\hat Z)}\\
&\ge \left(\frac{s_X s_Z}{2\pi e \, c}\right)^2\label{later}  ,
\end{align}
as claimed in Eq.~(\ref{variance}), where the last line follows as per the proof of the theorem in the main text.  

For example, take $X$ and $Z$ to be the usual qubit polarisation observables corresponding to Pauli operators $\sigma_X$ and $\sigma_Z$.  Then $s_X=s_Z=2$ and $c=1/2$, yielding
\begin{equation} \label{xyent}
\left[V_{\N}^X+\frac{1}{3} \right]\,\left[ V_{\D}^Z +\frac{1}{3} \right] \ge\frac{16}{\pi^2 e^2} \approx 0.219. 
\end{equation}
Hence, the noise and disturbance, as characterised by $V_{\N}^X$ and $ V_{\D}^Z $, can never both vanish, since $1/9\approx 0.111$.  Even stronger relations are possible in this regard.  For example, $V_{\N}^X=0$ implies that $p(m,x)=p(m)\delta_{x,f(m)}$, and hence that $H(X|M)=0$.  Since $H(Z|\hat{Z})\le \log 2=1$ for qubit observables, the information-theoretic noise-disturbance relation of Theorem~\ref{theo:main} then implies that $H(Z|\hat{Z})=1$, which is turn is only possible if $p(\hat z,z)=p(\hat z)/2$.  It follows that the MSD disturbance $V_{\D}^Z$ takes its maximum possible value, i.e., 
\begin{equation}  
V_{\D}^Z = 2\qquad{\rm for}\qquad V_{\N}^X=0. 
\end{equation}

The above tradeoff relations are quite strong in comparison to related tradeoff relations in \cite{Oz03,B13}.
To see this, note first that the root-mean-square deviations correspond to the conventional definitions of noise and disturbance for the case of a maximally-mixed system state. In particular, the POVM $\{E_m\}$ corresponding to the measurement outcomes of the apparatus $\mM$ follows from 
\[ p(m) = \Tr[\mM^m(\rho_S)] = \Tr[\rho_S (\mM^m)^*(\openone_S)] = \Tr[\rho_S E_m] \]
as $E_m=(\mM^m)^*(\openone_S)$, where $\phi^*$ denotes the dual of linear map $\phi$.  It follows that $p(x,m)$ in Definition 1 of the main text can be rewritten as
\[ p(m,x) = \frac 1d\Tr[\mM^m(|\psi^x\>\<\psi^x|)] = \frac 1d\Tr[|\psi^x\>\<\psi^x| (\mM^m)^*(\openone_S) = \frac 1d\Tr[|\psi^x\>\<\psi^x| E_m]. \]
Hence, one has 
\begin{eqnarray} \nonumber
V_{\N}^X &=& \sum_{m,x} p(m,x) [\hat x - \xi_x]^2 = \sum_{m,x} \frac{1}{d} \Tr[E_m |\psi^x\>\<\psi^x| [(f(m)-\xi_x)]^2\\
&=& \frac 1d \sum_m \Tr[E_m [f(m)-X]^2] = \epsilon(X,d^{-1} \openone_S)^2,
\end{eqnarray}
where $\epsilon(X,\rho_S)$ is the measure of noise proposed by Ozawa \cite{Oz03,Oz04}.  Similarly, it may be shown that
\begin{equation}  
V_{\D}^Z = \eta(Z,d^{-1} \openone_S)^2, 
\end{equation}
where $\eta(Z,\rho_S)$ is the measure of disturbance proposed by Ozawa. (We note that, even though these quantities are equivalent in this particular case, their physical interpretations remain different.)

It follows from the above results that, for example, for qubit polarisations $X$ and $Z$
\begin{equation}
\eta(Z,d^{-1} \openone_S) =\sqrt{2} \qquad{\rm for}\qquad \epsilon(X,d^{-1} \openone_S)=0.
\end{equation}
In contrast, known disturbance relations for $\epsilon(X,\rho_S)$ and $\eta(Z,\rho_S)$ place {\it no} restrictions on these quantities for the maximally mixed state $\rho_S=d^{-1}\openone_S$.

\section{Position-Momentum case}

Here we show how to derive the noise-disturbance uncertainty relations 
\begin{equation} \label{pqrelations}
	\N(\mM,Q) + \D(\mM,P) \ge \log \pi e\hbar,~~~~~V_{\N}^Q\,V_{\D}^P \ge \hbar^2/4
\end{equation}
in the main text, for conjugate position and momentum observables $Q$ and $P$. The {\it joint-measurement} uncertainty relations
\[ \N(\mM,Q) + \N(\mM,P) \ge \log \pi e\hbar,~~~~~V_{\N}^Q\,V_{\N}^P \ge \hbar^2/4 ,\]
analogous to Eq.~(\ref{eq:intrinsic-Hall}) of the main text, can be similarly derived for any joint measurement $\mM$ of $Q$ and $P$.

We first define an information-theoretic noise,  $\mN(\mM,Q)$, via a modification of the experiment in Fig.~\ref{fig:noise} of the main text.  In particular, a Gaussian state $|\psi^q\rangle$, with mean position $\langle Q\rangle=q$ and fixed position variance $(\Delta Q)^2=v$, is input to the measuring apparatus $\mM$ with  prior probability density $\wp(q)$ (specified below).  As $v\to0$ these input states approach eigenkets of $Q$, and so a perfect measurement of $Q$ by the apparatus can perfectly  distinguish the input states in this limit. Let $\wp(m|\psi^q)$ denote the conditional probability density for outcome $m$ given input state $|\psi^q\rangle$.  Hence, the corresponding information-theoretic noise is suitably quantified by the limiting conditional differential entropy
\begin{equation} \label{qnoise}
	\N(\mM,Q) := \lim_{v\to 0} H(Q|M) = \lim_{v\to 0} H(QM)-H(M),
\end{equation}
generalising Definition 1 of the main text.  
In particular, $\N(\mM,Q)$ is a measure of the average uncertainty of $Q$, given the outcome of the estimate, and approaches its minimum possible value ($-\infty$) in the limit that $\mM$ is a perfect position measurement.  Note that the corresponding entropic length, $2^{\N(\mM,Q)}$, is a more direct measure of the residual uncertainty of $Q$, having units of position and a minimum value of 0 in the limit of a perfect position measurement \cite{H99}.  Note also that the above definition is operational to the extent that for a `binned' position measurement $\mM$, with resolution $\delta q$, we have $\N(\mM,Q)\approx H(Q|M)\approx \log \delta q$ for finite $v\ll \delta q$.

For $\N(\mM,Q)$ to be uniquely defined, the prior density $\wp(q)$ must be specified to allow calculation of $H(Q|M)$.  To characterise the performance of the measurement over all possible values of $Q$, without bias, this density should become uniform in a suitable limit (analogously to the uniform distribution over eigenstates in Definition 1).  We therefore choose
\begin{equation} \label{pq}
	\wp(q) := (2\pi \bar v)^{-1/2} e^{-q^2/2\bar v}
\end{equation}
where $\bar v$ is specified further below and approaches $\infty$ in the limit of interest.

The information-theoretic disturbance, $\D(\mM,P)$, is similarly defined via a modification of the experiment in Fig.~\ref{fig:disturbance} of the main text.
In particular, a Gaussian state $|\phi^p\>$, with mean momentum $p$ and fixed momentum variance $(\Delta P)^2=w$, is input to the apparatus with prior probability density $\tilde \wp(p)$; subjected to an error correction process; and estimated via an ideal measurement of $P$.  Denoting the optimal estimate by $\hat P$, the corresponding information-theoretic disturbance is then
\begin{equation} \label{pdist}
	\D(\mM,P) := \lim_{w\to 0} H(P|\hat P) .
\end{equation}
This measure quantifies  the residual uncertainty of the momentum following the estimate, analogously to Eq.~(\ref{qnoise}), and clearly generalises Definition 2 of the main text.  A suitable prior density $\tilde \wp(p)$ is 
\begin{equation} \label{pp}
	\tilde \wp(p) := (2\pi \bar w)^{-1/2} e^{-p^2/2\bar w} ,
\end{equation}
where $\bar w$ is specified further below and approaches $\infty$ in the limit of interest ($v,w\to 0$), ensuring a uniform distribution in this limit. 

It is of interest to note that the two experiments above could be carried out simultaneously if the values of $v$ and $w$ were such that $vw= \hbar^2/4$, as in this case one could use minimum-uncertainty Gaussian states as input states, having position and momentum variances $v$ and $w$ respectively.  They could  also be carried out simultaneously for $vw>\hbar^2/4$, by using suitable mixed Gaussian input states rather than pure states.  Hence, the two experiments will be `complementary' whenever 
\begin{equation} \label{complement}
	v\, w< \frac{\hbar^2}{4},
\end{equation}
and in particular in the limit of interest, $v,w\to 0$.  This limit may be interpreted as corresponding to the input of uniform distributions of eigenkets of $Q$ and $P$, in direct analogy to the discrete case.

To obtain the noise-disturbance uncertainty relations (\ref{pqrelations}),  we follow a procedure similar to the earlier proof of the main Theorem.  The main differences are that (i) the joint state $|\Phi^+_{RS}\>$ is replaced by a two-particle Gaussian state, that approaches an Einstein-Podolsky-Rosen state as $v,w\to 0$ (i.e., an eigenket of relative position and total momentum);  and (ii) the Maassen-Uffink relation is replaced by a suitable entropic uncertainty relation for position and momentum.

In particular, introducing an auxilary reference system $R$, consider the two-system Gaussian state $|\Psi_{RS}\rangle$ with position representation
\[ \Psi_{RS}(x,y) := K e^{-\frac{(x-y)^2}{4f} }\,e^{-\frac{g(x+y)^2}{4\hbar^2} }  \]
for some $f,g>0$ (to be specifed below), where $K$ is a normalisation constant.  This state satisfies
\[ \< Q_R - Q_S\> = 0 = \< P_R + P_S\>,~~~~~{\rm Var}(Q_R-Q_S) = f,~~~~~{\rm Var}(P_R+P_S) = g, \]
and so approaches an eigenstate of relative position and total momentum in the limit $f,g\to0$.
It is straighforward to calculate that jointly measuring $\mM$ on $S$ and $Q^\star:=\alpha\, Q_R$ on $R$, with
\[ \alpha:= \frac{1+\frac{fg}{\hbar^2}}{1-\frac{fg}{\hbar^2} },\]
yields the same joint probability density, for outcomes $M=m$ and $Q^\star=q$,  as the first experiment above, i.e., 
\begin{equation} \label{pmq}
	p_{\Psi_{RS}}(m,q) = \wp(m,q) :=\wp(q)\,\wp(m|\psi^q) ,
\end{equation}
provided that $f$, $g$ and $\bar{v}$ are chosen to satisfy
\begin{equation}  \label{barv}
	v = \frac{f}{1+\frac{fg}{\hbar^2}},~~~~\bar v = \alpha^2 \frac{\hbar^2}{g} \left[ 1+\frac{fg}{\hbar^2}\right] .
\end{equation}
Similarly, jointly measuring $\hat Z$ on $S$ and $P^\star:=-\alpha P_R$ on $R$ yields the same joint probability density for outcomes $\hat P=\hat p$ and $P^\star=p$ as the second experiment above, i.e., 
\begin{equation} \label{ppp}
	p_{\Psi_{RS}}(\hat p,p) = \wp(\hat p, p) := \tilde \wp(p)\,\wp(\hat p|\phi^p) ,
\end{equation}
provided that $f$, $g$ and $\bar{w}$ are chosen such that
\begin{equation} \label{barw}
	w = \frac{g}{1+\frac{fg}{\hbar^2}},~~~~~~\bar w = \alpha^2 \frac{\hbar^2}{2f} \left[ 1+\frac{fg}{\hbar^2}\right] .
\end{equation}
The observables  $Q^\star$ and $P^\star$ defined above will play roles similar to that of $X^T$ and $Z^T$ in the proof of Theorem 1.

Equations~(\ref{barv}) and (\ref{barw}) can be explicitly solved to give $\bar v$, $\bar w$, $f$, $g$ and $\alpha$ in terms of $v$ and $w$, providing that $vw\leq \hbar^2/2$ (which of course holds in the limits of interest, $v,w\to0$, and is also guaranteed for finite $v$ and $w$ if the experiments are `complementary' in the sense of Eq.~(\ref{complement})).  As can already be seen from these equations, one finds $f,g\to 0$, $\bar v,\bar w\to\infty$ and $\alpha\to1$ in the limit $v,w\to 0$.  Explicitly, one has, for example, 
\[ f = \frac{v}{\sqrt{1-vw/\hbar^2}}, ~~~~g = \frac{w}{\sqrt{1-vw/\hbar^2}}, ~~~~~\alpha =\frac{1}{1-2vw/\hbar^2}, ~~~~\bar v = \frac{\hbar^2}{w} (1-2vw/\hbar^2)^{-2}(1-vw/\hbar^2)^{-1/2}. \]

Similarly to the proof of Theorem 1, $\mM$ and $\hat P$ correspond to a simultaneous measurement of some joint observable $U$ on the system, with outcome $u\equiv(m,\hat p)$ and POVM $\{\Pi^u\}$.  Hence, defining an ensemble of reference states ${\cal U}\equiv \{\rho^{u}_R; \wp(u)\}$ as before (with $|\Phi^+_{RS}\>$ replaced by $|\Psi_{RS}\>$), we have the corresponding chain
\[ H(Q^\star|U) + H(P^\star|U) = H(\alpha Q_R|U) + H(-\alpha P_R|U) = \log \alpha^2 + H(Q_R|U) + H(P_R|U) \geq \log \alpha^2\pi e\hbar ,\]
using the usual entropic uncertainty for position and momentum \cite{BM75}. But from Eqs.~(\ref{pmq}) and (\ref{ppp}) and the classical data-processing inequality, one also has $H(Q^\star|U)=H(Q^\star|M,\hat P)\leq H(Q^\star|M)=H(Q|M)$ and $H(P^\star|U)=H(P^\star|M,\hat P)\leq H(P^\star|\hat P) = H(P|\hat P)$. Hence,
\begin{equation} \label{last}
	H(Q|M) + H(P|\hat P) \geq \log \alpha^2 \pi e\hbar .    
\end{equation}
Recalling that $\alpha\to 1$ in the limit $v,w\to 0$, we immediately obtain the first relation in Eq.~(\ref{pqrelations}) via definitions (\ref{qnoise}) and (\ref{pdist}), as desired.  

Finally,  the second relation in Eq.~(\ref{pqrelations}) is obtained by following essentially the same steps used in the proof of Eq.~(\ref{later}) above (see also the derivation of the Heisenberg inequality in Ref.~\cite{BM75}), to show that $\< (\hat q-q)^2\> \geq (2\pi e )^{-1} 2^{2H(Q|M)}$ and   $\< (\hat p-p)^2\> \geq (2\pi e )^{-1} 2^{2H(P|\hat P)}$, and again using Eq.~(\ref{last}) in the limit $v,w\to0$.  

\end{document}